\documentclass[10pt,sigconf,nonacm]{acmart}

\usepackage{xcolor}
\usepackage{framed}
\renewenvironment{shaded}{%
  \MakeFramed{\advance\hsize-\width \FrameRestore\FrameRestore}}%
 {\endMakeFramed}
\definecolor{shadecolor}{gray}{0.75}
\usepackage{amsmath}
\usepackage{filecontents}
\usepackage[utf8]{inputenc}
\usepackage{subfigure}

\usepackage{booktabs,color,doi,graphicx,latexsym,url,xcolor,xspace}
\usepackage{booktabs,color,doi,graphicx,latexsym,url,xcolor,xspace}
\usepackage{cleveref}
\usepackage{microtype,hyperref}
\usepackage{eucal}
\usepackage[inline]{enumitem}
\definecolor{citecolor}{HTML}{0000C0}
\definecolor{urlcolor}{HTML}{000080}

\usepackage{tikz}
\usepackage{varwidth}
\usepackage{balance}
\usepackage{multirow}
\usepackage{colortbl}
\usepackage{hhline}
\usepackage[multiple,para]{footmisc}
\usepackage{array}

\usepackage{cleveref}
\definecolor{lgray}{rgb}{0.95,0.95,0.95}
\hypersetup{
    colorlinks=true,
    linkcolor=black,
    citecolor=citecolor,
    filecolor=black,
    urlcolor=urlcolor,
}

\usepackage{url}
\usepackage[ruled,vlined,linesnumbered]{algorithm2e}
\usepackage{balance}
\usepackage{xspace}
\usepackage{todonotes}
\usepackage{hyperref}
\usepackage{enumitem}

\newcommand{\system}{{\sc{Kevin}}\xspace}

\def\update{\rho}
\def\reconf{\delta}
\def\resrv{r}

\DeclareMathOperator*{\dist}{dist_{DB}}
\newcommand{\da}{\texttt{DA}\xspace}
\newcommand{\debruijn}{de Bruijn\xspace}

\newtheorem{observation}{Observation}
\newtheorem{claim}{Claim}

\title{Kevin: de Bruijn-based topology with demand-aware links and greedy routing}

\author{Johannes Zerwas}
\affiliation{%
  \institution{TU Munich}
  \country{Germany}
}

 \author{Csaba Gy{\"o}rgyi}
 \affiliation{%
   University of Vienna and
   ELTE Eötvös Loránd University
   \country{Austria and Hungary}
}

 \author{Andreas Blenk}
 \affiliation{%
   \institution{TU Munich}
   \country{Germany}
  }

 \author{Stefan Schmid}
 \affiliation{%
   \institution{TU Berlin and University of Vienna}
   \country{Germany and Austria}
}

 \author{Chen Avin}
 \affiliation{%
   \institution{Ben-Gurion University}
   \country{Israel}
}

\date{}
\begin{document}

\sloppy 

\begin{abstract}
We propose \system, a novel demand-aware reconfigurable rack-to-rack datacenter
network realized with a simple and efficient control plane. In particular, \system~makes
effective use of the network capacity by supporting integrated and multi-hop  
routing as well as work-conserving scheduling. To this end, \system~relies on local greedy routing with small forwarding tables
which require local updates only 
during topological reconfigurations, making this approach ideal for dynamic networks. 
Specifically, 
\system is based on a de Bruijn topology (using a small number of optical circuit switches) in which static links are enhanced with 
opportunistic links.
\end{abstract}
\settopmatter{printfolios=true}

\maketitle

\section{Introduction}

The performance of many cloud applications, e.g., related to distributed machine learning, batch processing, or streaming, critically depends on the bandwidth capacity of the underlying network.
High network throughput requirements are also introduced by today's trend of resource disaggregation in datacenters, where fast access to remote resources (e.g., GPUs or memory) is critical for the overall system performance~\cite{talk-about,li2019hpcc}. Accordingly, over the last years, great efforts have been made to improve the throughput of datacenter networks~\cite{bcube,singla2012jellyfish,jupiter,AlFares2008}.

Emerging optical technologies enable a particularly innovative approach to improve datacenter performance, by supporting dynamic reconfigurations of the physical network topology~\cite{ballani2020sirius,zhou2012mirror,kandula2009flyways,rotornet,opera,helios,firefly,megaswitch,quartz,osa,projector,cthrough,splaynets,venkatakrishnan2018costly,schwartz2019online,proteus,100times,fleet,zhang2021gemini}. In particular, optical circuit switches allow to provide dynamic connectivity in the form of \emph{matchings}~\cite{sigmetrics22cerberus,infocom22lazy}. Reconfigurable datacenter networks (RDCNs) use such switches to establish topological shortcuts (i.e., shorter paths) between racks, hence utilizing available bandwidth capacity more efficiently and improving throughput~\cite{rotornet,opera,sigmetrics22cerberus}.

Reconfigurable datacenter networks come in two flavors: oblivious and demand-aware~\cite{ccr18san,osn21}. Oblivious RDCNs such as RotorNet~\cite{rotornet}, Opera~\cite{opera}, and Sirius~\cite{sirius}, rely on quickly and periodically changing interconnects between racks, to emulate a complete graph. Such emulation was shown to provide high throughput and is particularly well-suited for all-to-all traffic patterns~\cite{sigmetrics22cerberus}.
In contrast, demand-aware RDCNs allow to \emph{optimize} the topological shortcuts, depending on the traffic pattern. Demand-aware networks such as ProjecToR~\cite{projector}, Gemini~\cite{zhang2021gemini}, or Cerberus~\cite{sigmetrics22cerberus}, among many others~\cite{zhou2012mirror,kandula2009flyways,firefly,osa,projector,100times,fleet,dan,flexspander}, are attractive since datacenter traffic typically features much temporal and spatial structure: traffic is bursty and skewed, and a large fraction of communicated bytes belong to a small number of elephant flows~\cite{tracecomplexity,benson2010network,projector,osa,datacenter_burstiness,DBLP:journals/cn/ZouW0HCLXH14}. By adjusting the datacenter topology to support such flows, e.g., by providing direct connectivity between intensively communicating source and destination racks, network throughput can be increased further (even if done infrequently~\cite{zhang2021gemini}). 

However, the operation of reconfigurable datacenter networks comes with overheads and limitations. 
In general, existing RDCNs typically rely on a hybrid topology which combines static (electrical) and dynamic (optical) parts. While such a combination is often very powerful~\cite{projector}, current architectures support only  fairly restricted routing. First, communication on the (dynamic) optical topology is often limited to a one or two hops, constraining the possible path diversity,  and hence capacity, of the optical network~\cite{ancs18,projector,sirius,zhang2021gemini,sigmetrics22cerberus}. Furthermore, routing is usually \emph{segregated}: flows are either only forwarded along the static or the dynamic network, but not a combination of both~\cite{projector,sigmetrics22cerberus,taleoftwo}.
The restriction to segregated routing also entails overheads as it requires significant buffering while the reconfigurable links are not available. As static links are always available for packet forwarding in hybrid datacenter networks, this segregation is also not work conserving.
Demand-aware RDCNs may introduce additional delays as they require potentially time-consuming optimizations. 

This paper is motivated by the desire to overcome these limitations, and 
to better exploit the available link resources, by supporting a general \emph{multi-hop} and \emph{integrated} (i.e., non-segregated) routing. Specifically, we envision a datacenter network in which packets can be forwarded in a work-conserving manner, along \emph{any} available link, be it static or dynamic, and in which a routing path can combine both link types.   
Such an integrated and work-conserving routing also has the potential to avoid long buffering times and hence delays: if a reconfigurable link is currently unavailable, packets can directly be forwarded to the other available (static) links. 
However, going beyond segregated and 1- or 2-hop routing, requires a novel network control plane: traditional routing protocols based on shortest paths are not designed for highly dynamic topologies and the frequent recomputation of routes can become infeasible~\cite{francois2005achieving}. Furthermore, to keep update cost low and provide a high scalability, it is desirable to have small forwarding tables.  

\begin{table}[]
    \vspace{-10mm}
    \footnotesize
    \centering
    \begin{tabular}{|m{1.2cm}|>{\centering\arraybackslash}m{.78cm}|>{\centering\arraybackslash}m{.78cm}|>{\centering\arraybackslash}m{.78cm}|>{\centering\arraybackslash}m{.79cm}|>{\centering\arraybackslash}m{.78cm}||>{\centering\arraybackslash}m{.78cm}|}
         \multicolumn{1}{c}{} & \multicolumn{1}{c}{\rotatebox{90}{Xpander~\cite{xpander}}} & \multicolumn{1}{c}{\rotatebox{90}{Sirius~\cite{sirius}}} & \multicolumn{1}{c}{\rotatebox{90}{ProjecToR~\cite{projector}}} & \multicolumn{1}{c}{\rotatebox{90}{Gemini~\cite{zhang2021gemini}}} & \multicolumn{1}{c}{\rotatebox{90}{Cerberus~\cite{sigmetrics22cerberus}}} & \multicolumn{1}{c}{\system} \\
    \hline
         Integrated Multi-hop & \cellcolor{green!15} Yes & \cellcolor{red!15} 2-hops & \cellcolor{red!15} 1 hop & \cellcolor{red!15} 2-hops & \cellcolor{red!15} 1-hop & \cellcolor{green!15} \bf Yes \\
    \hline
         Demand-aware & \cellcolor{red!15} No & \cellcolor{red!15} No & \cellcolor{green!15} Yes & \cellcolor{green!15} Yes & \cellcolor{green!15} Yes  & \cellcolor{green!15} \bf Yes \\
    \hline
         Work \newline Conserving & \cellcolor{green!15} Yes & \cellcolor{red!15} No & \cellcolor{red!15} No & \cellcolor{green!15} Yes & \cellcolor{red!15} No & \cellcolor{green!15} \bf Yes \\
    \hline
         Topology Update & \cellcolor{green!15} None & \cellcolor{green!15} Fast & \cellcolor{green!15} Fast & \cellcolor{red!15} Slow & \cellcolor{green!15} Fast  & \cellcolor{green!15} \bf Fast \\
    \hline
         Routing \& Control & \cellcolor{green!15} Simple & TBD & TBD & \cellcolor{green!15} Simple & TBD & \cellcolor{green!15} \bf Simple\\
    \hline
    \end{tabular}
    \caption{Recent (R)DCN designs and their properties.}
    \label{tab:recent}
    \vspace{-5mm}
\end{table}

To this end, we propose a simpler and more efficient control plane for RDCNs which avoids flow forwarding delays by supporting
\emph{local and greedy integrated routing}:
the forwarding rules depend on local information only, i.e., the set of direct neighbors as well as information in the packet header (in particular, the destination); they are hence not affected by topological changes in other parts in the network and do not have to be updated under reconfigurations. This can significantly reduce control plane overheads during topological adjustments, maintaining a simple routing and control, and is hence well-suited for highly dynamic networks. As we will show, the greedy routing approach is also compact and only requires small forwarding tables.  

In particular, we present \system, a novel demand-aware reconfigurable datacenter
network which leverages such a local control plane using a \debruijn topology (built from a small number of optical switches), 
in which static links are enhanced with 
opportunistic links.
\system uses logical addressing, and is based on a receiver-based approach for
the efficient detection of elephant flows as well as the local and collision-free scheduling
of demand-aware links.
The control plane of \system~can be realized both using centralized or distributed algorithms. In both cases, 
it reduces buffer requirements and delays, supports very small forwarding tables based on standard longest common prefix matching,  and
enables fast and local route updates as well as short path lengths.
\system is well-suited to be realized using the Sirius~\cite{sirius} architecture.   
In summary, we make the following \textbf{\emph{contribution}}:

\begin{itemize}[leftmargin=*]
  \item We present \system, a novel and pratical reconfigurable datacenter architecture which supports efficient multi-hop, integrated and work-conserving routing, to push the performance limits in datacenter networks.  
  The simplicity of \system relies on the observation that adding shortcuts to a static \debruijn topology allows to continue supporting greedy local routing. 
\end{itemize}

%\noindent
%\textit{Note: This work does not raise any ethical issues.}

\section{Putting \system into Perspective}\label{sec:motivation}

To put \system into perspective with the most recent proposals for datacenter network designs, we use Table~\ref{tab:recent}.
We first consider Xpander~\cite{xpander}, a state-of-the-art \emph{static and demand-oblivious} topology, which is based on expander graphs. While Xpander has many attractive properties, according to recent results (including the ones in this work), we expect that reconfigurable datacenter networks (i.e., based on dynamic topologies) can provide an improved performance.

We next consider Sirius~\cite{sirius}, a recent proposal by authors from Microsoft, as an example of a \emph{dynamic and demand-oblivious} topology (similar to~\cite{rotornet,opera}). Such topologies  have been shown to be very effective as well, but still have several potential deficits. First and foremost, they do not feature demand-aware links: while the role and use of demand-aware topology components in future datacenters is generally still subject to ongoing discussions, empirical studies show that demand-awareness can improve throughput under today's typical skewed workload distributions~\cite{sigmetrics22cerberus,zhang2021gemini}.
Furthermore, current dynamic and demand-oblivious designs are limited to at most 2-hop routing on dynamic links, and are not work conserving. There also remain some open questions regarding the complexity of the control and routing of these systems. 

Then there are also systems which are dynamic and demand-aware. Systems like ProjecToR~\cite{projector} use a combination of demand-aware optical and electric switches, but do not support integrated mutli-hop routing (ProjecToR uses only 1-hop on demand-aware links), are not work conserving, and the control complexity is not fully determined.  Gemini~\cite{zhang2021gemini}, a recent proposal by authors from Google, makes the case for demand-aware links in production level datacenters, but it currently implements only infrequent topology updates (about once a day). Lastly, Cerberus relies on a combination of three topologies: static, dynamic oblivious, and dynamic demand-aware, which together, potentially provide higher throughput. However, it supports only 1-hop routing on the demand-aware links, and its control and routing mechanisms are left abstract and require further investigation.

\begin{shaded}
In contrast to all the above systems, \system features all the desired properties  listed in Table \ref{tab:recent}. It supports \emph{integrated multi-hop} routing as in Xpander, it uses \emph{demand-aware} links as ProjecToR, it is \emph{work conserving} as Gemini, it enables fast \emph{topology updates} as in Cerberus, and it is based on simple control and routing.
\end{shaded}

To put the novelty of our contribution into perspective, we note that the overheads and limitations of shortest path routing has already been studied in several contexts, including dynamic and mobile networks such as ad-hoc networks where greedy approaches such as geo-routing can be an attractive alternative~\cite{stojmenovic2002position}. 
We are also not the first to observe the benefits of \debruijn based networks in dynamic settings, and there exist peer-to-peer~\cite{scheideler2009distributed,naor2007novel,fraigniaud2006d2b,kaashoek2003koorde} and parallel architectures~\cite{louri1995optical} which rely on \debruijn graphs.
However, to the best of our knowledge, we are the first to study such an approach in the context of reconfigurable and demand-aware datacenter networks. 

\begin{figure}
    \centering
    \includegraphics[width=\linewidth]{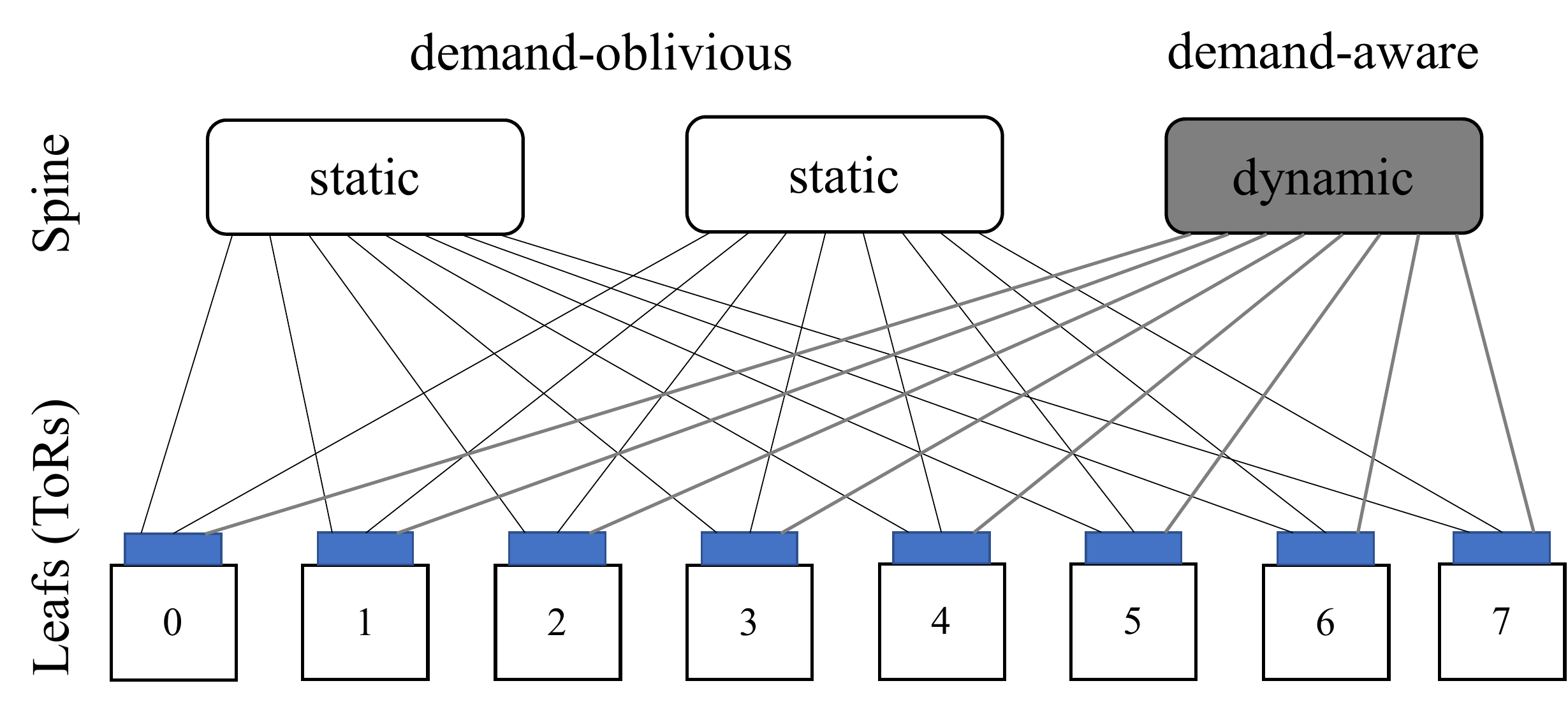}
    \caption{TMT network example with eight ToR switches and three spine switches from which two are static matchings and one is dynamic matchings.
    }
    \label{fig:deisgn:tmt_network}
    \vspace{-0.3cm}
\end{figure}

\section{The \system RDCN}\label{sec:duo}

The rack-to-rack network provided by \system~is based on the
 ToR-Matching-ToR (TMT) model~\cite{rotornet,sigmetrics22cerberus}: $n$ top-of-rack (leaf) switches are interconnected by a set of $k$ optical spine switches. Each spine switch provides a $n \times n$ directed matching between its input-output ports.
 Depending on the switch type, the matching can dynamically change over time.
In particular, \system~is hybrid, in the sense that one part of
the topology model is demand-oblivious and static using $k_s$ spine switches (i.e., static matchings),  and the other part is
dynamic and demand-aware using $k_d$ reconfigurable spine switches (i.e., dynamic matchings), and $k=k_s+k_d$.
Figure~\ref{fig:deisgn:tmt_network}, presents an example of the TMT model with eight ToR switches and three spine switches, from which two are static and one is dynamic.
Each ToR-spine link in the figure represents one directed uplink and one directed downlink. It is important to note that abstractly, we use $k$ spine switches, 
each implementing an $n \times n$ matching, but each matching can be split across a set of several smaller switches, like in Sirius~\cite{sirius}.

In order to maximize performance, 
\system~uses dynamic, demand-aware links to provide shorter paths for elephant flows, while other flows are transmitted via the combined (static + dynamic) topology.
A key feature of \system~is that it supports \emph{integrated multi-hop} routing
across \emph{both} switch types. This is in contrast to previous works that rely on \emph{segregated} and single-hop forwarding for demand-aware links~\cite{projector,sigmetrics22cerberus}. Moreover routing in \system\ is \emph{efficient}, by relying on logical addressing and a local control plane, implementing greedy routing. Thus, links can always be used immediately, with a work-conserving scheduler. 
To detect elephant flows, \system~leverages a simple sketch, sampling the flow sizes
and then adjusting flow paths accordingly. 
In the following, we present the different components of \system~in detail.

\subsection{The Hybrid Topology}

\system~combines two topologies, a static, demand-oblivious topology (the ``\emph{backbone}''), and a dynamic, opportunistic, demand-aware topology into a unified one. 
Both topologies are built from matching switches according to the TMT model, forming a augmented \debruijn network~\cite{deBruijn}.

\begin{itemize}[leftmargin=*]
 \item \textbf{Static and demand-oblivious \debruijn topology (backbone):} 
 The static topology of \system~relies on a \debruijn graph. 
 It is formed by $k_s$ static optical circuit switches or patch panels. 
 \item \textbf{Demand-aware topology:} The static topology is enhanced
 by $k_d$ reconfigurable matchings, also implemented with optical circuit switches.
 The demand-aware (\da) links add \emph{shortcuts} on top of the static \debruijn topology.
\end{itemize}

\begin{figure}
    \centering
    \includegraphics[width=0.9\linewidth]{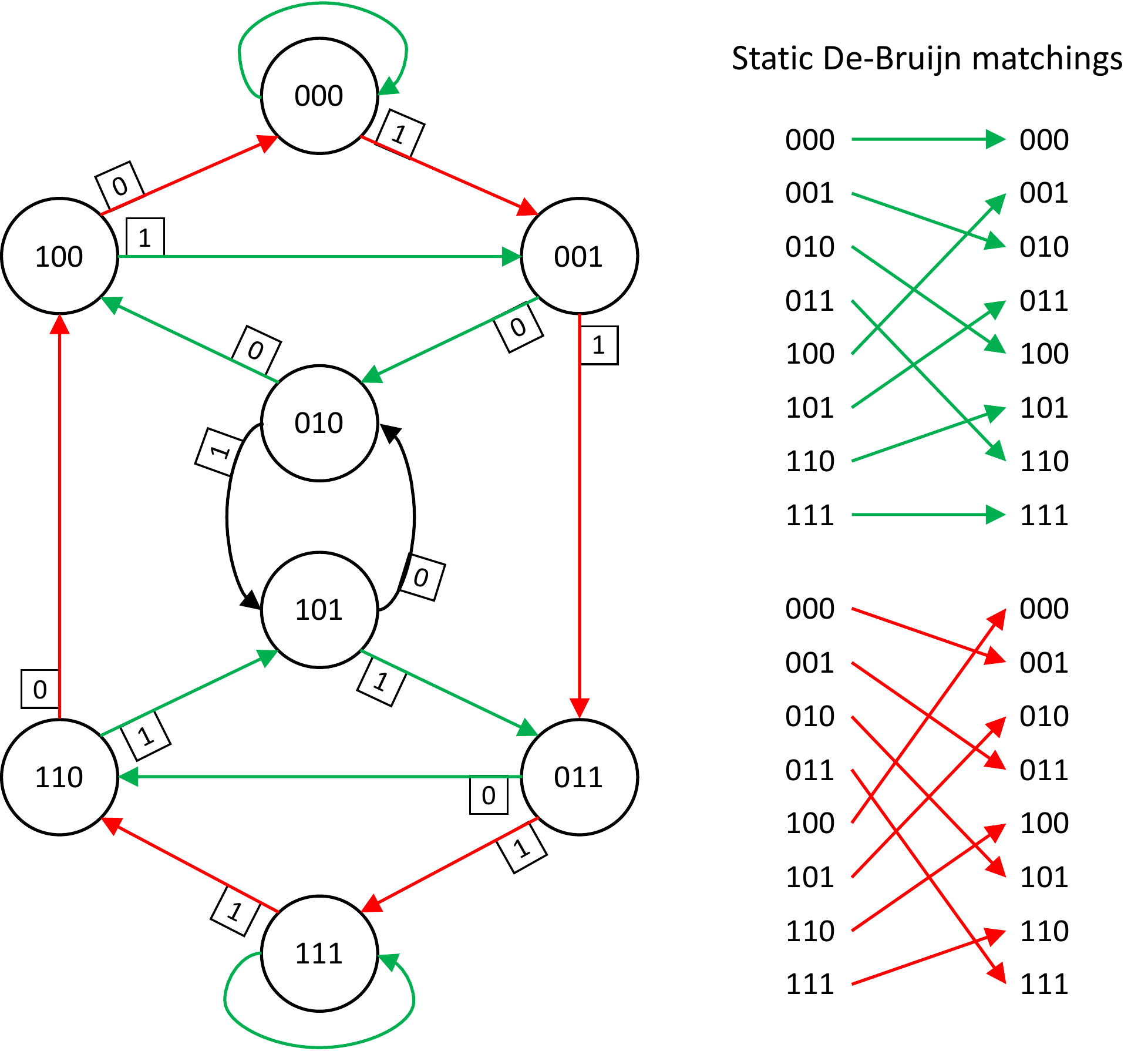}
    \caption{A $DB(2,3)$ static \& directed \debruijn graph with eight ToRs and its two corresponding matchings (colored in green and red). Each port (edge) is labeled $0$ or $1$ according to the performed shift operation.} 
    \label{fig:design:debruijn}
\end{figure}

We first discuss the static \debruijn topology, and how 
it can be built from a constant number of matchings (already two matchings suffice). 

\subsubsection{The Static \debruijn Topology}

We start with formally defining the \debruijn topology~\cite{leighton2014introduction}. For $i \in \mathbb{N}$, let $[i]=\{0, 1, \dots, i\}$.

\begin{definition}[\debruijn topology]
For integers $b,d >1$, the {\em $b$-ary \debruijn graph of dimension $d$}, $DB(b,d)$, is a directed graph $G=(V,E)$ with $n=b^d$ nodes and $m=b^{d+1}$ directed edges. The node set $V$ is defined as $V=\{v \in [b-1]^d \}$, 
i.e.,  $v=(v_1,\ldots,v_d), v_i \in [b-1]$, 
and the directed edge set $E$ is:
\begin{align}\label{eq:dedge}
   \{v,w\} \in E \Leftrightarrow w \in \{ (v_2,\ldots,v_{d},x): \; x \in [b-1] \}
\end{align}
\end{definition}

Note that the directed neighbors of node $v$ are determined by a left \emph{shift} operation on the address of $v$ and entering a new symbol  $x \in [b-1]$ as the 
right most (least significant) symbol.
It is well known that the \debruijn topology has the following properties:
\begin{enumerate}[leftmargin=*]
    \item \label{itm:reg} Considering self-loops, $DB(b,d)$ is a $b$-regular directed graph
    \item $DB(b,d)$ supports greedy routing with paths of length at most $d$
\end{enumerate}

The following observation will be relevant for our network design,  
it is a consequence of Property \eqref{itm:reg} above and Hall's theorem~\cite{hall1935}:
\begin{observation}\label{obs:union}
A $DB(b,d)$ topology can be constructed from the union of $b$ directed perfect matchings.
\end{observation}

Figure~\ref{fig:design:debruijn} demonstrates the $DB(2,3)$ \debruijn topology with $8=2^3$ nodes (ToRs) and two matchings (colored in green and red) that can be combined to create it.
Each node in the topology has two outgoing and two incoming directed links (including self loops). 
The figure shows the \emph{labeled} version of the graph where 
each edge (or a node outgoing port) is labeled with $0$ or $1$ according to the shift operation implied by Eq. \eqref{eq:dedge}.

It follows from Observation~\ref{obs:union} that we can build a 
$DB(k_s,d)$ topology with $k_s$ static spine switches.

\begin{algorithm}[t]
\caption{Building the DB Forwarding Table}
\label{alg:forwarding:merge}
\SetKwProg{merge}{Function \emph{BuildTable}}{}{end}
\merge{}{
    \For{each neighbor $z_1z_2,z_3$ at port $p$}
    {Add the following entries to the table
    \begin{tabular}{c|c|c}
        Prefix & Port & Path-length \\\hline
        $z_3**$ & $p$ & $3$  \\
        $z_2z_3*$ & $p$ & $2$  \\
        $z_1z_2z_3$ & $p$ & $1$  \\\hline
    \end{tabular}
    }
    Reduce the forwarding table according to \emph{LPM}
}
\vspace{-0.2cm}
\end{algorithm}

\begin{figure}[t]
\begin{tabular}{cc}
\begin{minipage}{0.45\columnwidth}
    \centering
    \begin{tabular}{c|c|c}
        Prefix & Port & Path-len \\\hline
        \multicolumn{3}{c}{(neighbor 110 on port 0)}\\
        $0**$ & $0$ & $3$  \\
        $10*$ & $0$ & $2$\\
        $110$ & $0$ & $1$ \\
        \multicolumn{3}{c}{(neighbor 111 on port 1)}\\
        $1**$ & $1$ & $3$  \\
        $11*$ & $1$ & $2$\\
        $111$ & $1$ & $1$
    \end{tabular}
\end{minipage}
&
\begin{minipage}{0.45\columnwidth}
\centering
    Reduced Table for ToR $011$
    \begin{tabular}{c|c|c}
        Prefix & Port & Path-len \\\hline
        $0**$ & $0$ & $3$  \\
        $10*$ & $0$ & $2$  \\
        $110$ & $0$ & $1$  \\
        $111$ & $1$ & $1$  \\\hline
        $011$ & Local & $0$
    \end{tabular}
\end{minipage} \\
(a) Neighbors' entries & (b) Reduced table
\end{tabular}
\caption{The results of building the forwarding table (Algorithm~\ref{alg:forwarding:merge}) of node 011 with neighbors 110 and 111 on the static $DB(2,3)$ \debruijn graph.}
\label{fig:design:example_tables}
\end{figure}

\begin{figure*}[t!]
\centering
\begin{tabular}{ccc}
\begin{minipage}{0.35\linewidth}
\includegraphics[width=\linewidth]{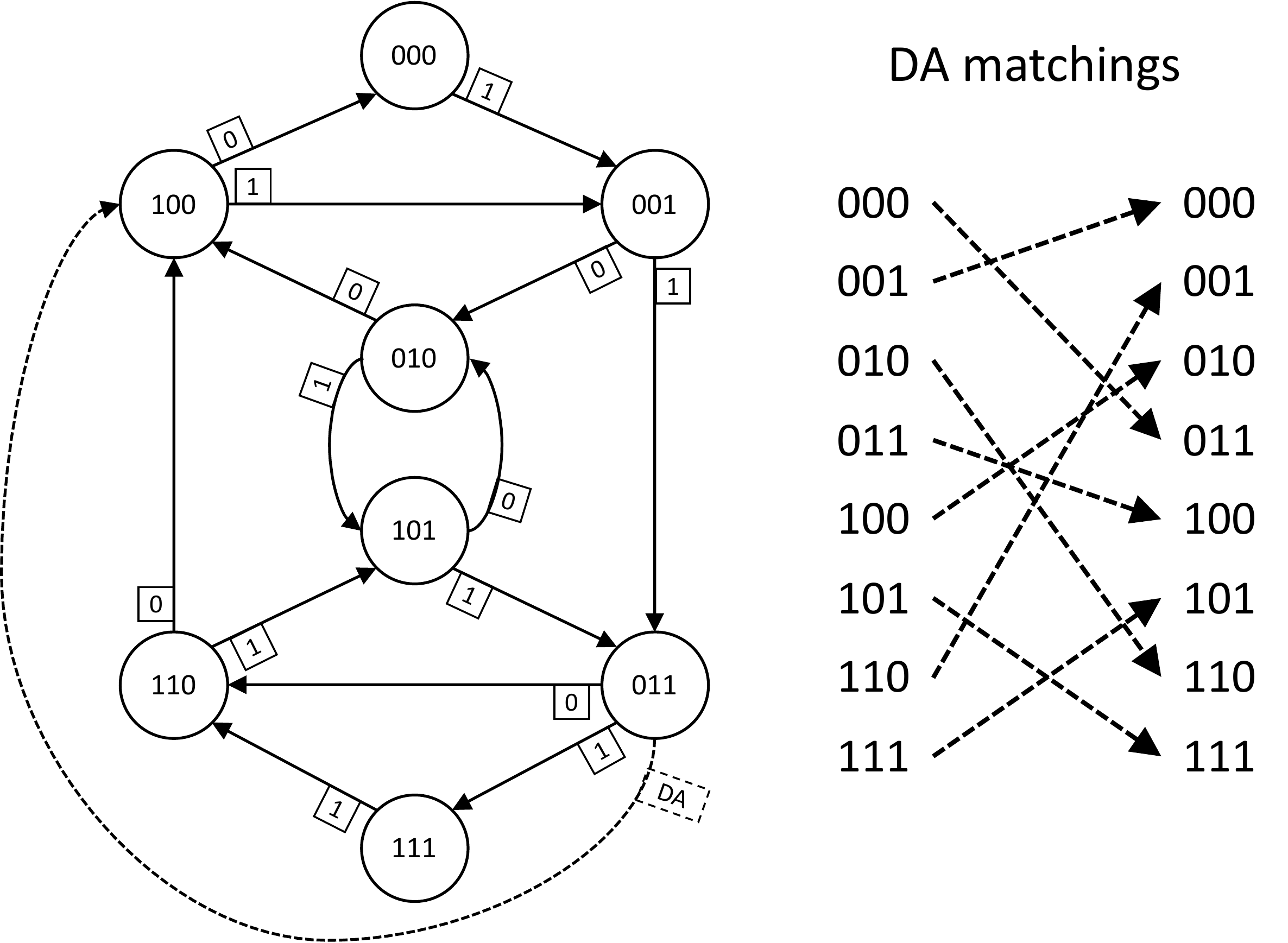}
\end{minipage}
&
\begin{minipage}{0.22\linewidth}
    \centering
    \begin{tabular}{c|c|c}
        Prefix & Port & Path-len \\\hline
        $0**$ & $DA$ & $3$  \\
        $00*$ & $DA$ & $2$\\
        $100$ & $DA$ & $1$
    \end{tabular}
\end{minipage}
&
\begin{minipage}{0.34\linewidth}
    \centering
    \begin{tabular}{l|c|c|l}
        Prefix & Port & Path-len & IP\\\hline
        $0**$ & $\{0, DA\}$ & $3$ & $10.0.0.0/9$ \\
        $00*$ & $DA$ & $2$ & $10.0.0.0/10$ \\
        $10*$ & $0$ & $2$ & $10.128.0.0/10$ \\
        $100$ & $DA$ & $1$ & $10.128.0.0/11$ \\
        $110$ & $0$ & $1$ & $10.192.0.0/11$ \\
        $111$ & $1$ & $1$ & $10.224.0.0/11$ \\
        \hline
        $011$ & Local & $0$ & $10.96.0.0/11$
    \end{tabular}
\end{minipage}\\
(a) ToR $011$ with new DA link to $100$. & (b) Entries from $100$. & (c) Reduced table on ToR $011$.
\end{tabular}
\caption{The new forwarding tables of ToR $011$ after the establishment of the \da-link from $011$ to $100$.}
\label{fig:design:dynamic_tables}
\end{figure*}

\subsubsection{Greedy and LPM Routing in \debruijn Topology}

It is well known that the \debruijn topology supports greedy routing from a source $s$ to a destination $t$ based solely on the address of $t$. That is, to choose the next-hop toward $t$ each node on the route needs to know the address of $t$ and the address of its neighbors. The next-hop is chosen as the neighbor which minimizes the \emph{\debruijn distance} to $t$.
The \debruijn distance between two nodes $v,w$ denoted as $\dist(v,w)$ is the minimum number of \emph{shift} operations needed to transform $v$ address to $w$ address. 
The main observation is  that each such shift implies a directed edge and the next-hop in the routing. For example, the \debruijn distance between node $s=011$ and $t=001$ is $\dist(s,t)=3$ and route from $s$ to $t$ in $DB(2,3)$ is $011 \rightarrow 110 \rightarrow 100 \rightarrow 001$ (see also Figure~\ref{fig:design:debruijn}). Note that in each hop the distance to $t$ is reducing.

A less-known fact is that routing on the \debruijn topology can 
be realized via a simple forwarding table that is based on a \emph{longest prefix match} (LPM)~\cite{durr2016flat}. 
To build the forwarding table for a node $v$ it only needs to know the address of each neighbor $w$ and the outgoing port $p$ that connects to it. 
Algorithm~\ref{alg:forwarding:merge}
describes the forwarding table building (for simplicity only for the $DB(2,3)$ case) and Figure~\ref{fig:design:example_tables} shows the forwarding table of node $011$ and how it is built from its neighbors $110$ and $111$. 
Note that rule $1**$, is removed from the table in the \emph{reduce} process since it will never be used.

Following Algorithm~\ref{alg:forwarding:merge}, we can state the following about the \emph{size} of the forwarding table of each node:
\begin{observation}
The longest prefix match forwarding table size of each node in a $DB(b,d)$ topology has at most $b d = O(b \log_b n)$ entries.
\end{observation}

We can now discuss the \da links and how they are merged into the hybrid topology.

\subsubsection{The Demand-aware Topology}

The simplicity of \system relies on the observation that adding shortcuts to the static \debruijn topology is easy and allows to continue supporting greedy and LPM routing. 
Recall that in our model we have $k_d$ switches or matchings for \da links. For now consider these $k_d n$ links as arbitrary links.
Later we discuss how to choose these links based on flow sizes.

Let $G=DB(b,d)$ be a \debruijn topology over the set $V$ of nodes.
Let $M$ be a directed matching on $V \times V$.
Let $H=G \cup M$ be the union of the directed graphs $G$ and $M$ with the same set of nodes $V$. We can claim the following about $H$.

\begin{claim}\label{clm:hybrid}
If we perform Algorithm~\ref{alg:forwarding:merge} on each node in $H$, then $H$ supports integrated, multi-hop, greedy, LPM routing 
with forwarding table size of $(b+1)d$.
\end{claim}
\begin{proof}[Proof sketch]
First we show that $H$ supports greedy routing, namely the next-hop is the neighbor with the shortest \debruijn distance to the destination. While greedy routing on the static topology reduces the distance function in each hop by exactly one, \da links can reduce it by more than one hop. From the greedy routing it is clear that LPM forwarding will work and that the path is integrated in a multi-hop manner.
\end{proof}

Figure~\ref{fig:design:dynamic_tables} demonstrates the $DB(2,3)$ topology with the addition of a single demand-aware matching (showing only one \da link from $011$ to $100$).
The figure also presents the new forwarding table at node $011$, which is constructed using Algorithm~\ref{alg:forwarding:merge}.
If we consider as before the route from $s=011$ to $t=001$ 
it will now be shorter $011 \rightarrow 100 \rightarrow 001$. 
In fact, all addresses with LPM $00*$ will use the new \da port for forwarding on node $011$. 
Note also that routes toward addresses with LPM $0**$, like $010$, have now two equal length routes (of length three), for example, $011 \rightarrow 100 \rightarrow 001 \rightarrow 010$ or $011 \rightarrow 110 \rightarrow 101 \rightarrow 010$.

Following Claim~\ref{clm:hybrid}, we can extend this example to more than one matching and support $k_d$ demand-aware matchings. 
Formally, for integers $k_s,k_d,x \ge 2$ and $n=(k_s)^x$, we denote by $\system(n,k_s,k_d)$ the \system topology with $k=k_s+k_d$ spine switches,  backbone network $DB(k_s ,\log_{k_s} n)$, and $k_d$ demand-aware switches.
We can state the following about the hybrid topology of \system.

\begin{theorem}\label{the:duo}
The $\system(n,k_s,k_d)$ topology with $n$ ToRs and $k=k_s + k_d$ spine switches ($k_d, k_s \ge 2$) supports integrated, multi-hop, greedy, LPM routing with forwarding table size of $O(k \log_{k_s} n)$ and diameter $d = O(\log_{k_s} n)$.
\end{theorem}

\subsection{Scheduling of Demand-Aware Links}
\label{subsec:da-sched}

\system relies on a control plane which can use 
centralized or decentralized scheduling of the \da links.
The centralized scheduling benefits from the global view, 
while the decentralized scheduling supports fast reaction.

We use Sirius'~\cite{sirius} reconfiguration model also for \da links: spine switches use passive gratings while (sending) ToR switches rely on tunable lasers which determine the link to set up in the corresponding switch (matching). 
This property is useful for the distributed version of the scheduling where the receivers provide permissions to senders to reconfigure links. 
All algorithms use the command `Set \da-link $(x,y,i)$' which means that sender ToR $x$ tunes its laser on port $i$ to establish a direct link to ToR $y$ via switch $i$.  Recall that each ToR has $k$ up-link toward the $k$ spine switches so we identify port $i$ with switch $i$.

\begin{algorithm}[t]
\caption{Centralized (BFS) \da links setting}
\label{alg:dalinks:centralized:BFS}
\SetKwProg{daconnected}{Function \emph{BFS-\da-links}($D$ - Demand Matrix, $k_d$ - number of \da switches)}{}{end}
\daconnected{} {
    $\Delta$=Largest $k_d n$ demands in $D$, sorted by volume \\
        \ForAll{$(s,t) \in \Delta$ from large to small} {
            $(x, switches) = $ FowardBruijn$(s,t)$ \label{ln:forward}\\
            $(y, i) = $ BackwardBruijnBFS$(s,t, switches)$ \label{ln:backward}\\
            \If{$\dist(s,x) + \dist(y,t) +1 \leq  \mathrm{dist}_{\system}(s,t)$}
            {
            \label{ln:ifBSF}
                Set \da-link $(x,y,i)$ \label{ln:setBFS}\\
            }
        }
}
\SetKwProg{daconnected}{Function \emph{FowardBruijn}$(s,d)$}{}{end}
\daconnected{} {
    $v=s$ \\
    \While{$v != t$}{ 
        \uIf{$v$ has available \da ports}{
            Return $(v, switches)$ \label{ln:foundForward}
        }
        \Else{
            $v =$ next-hop node toward $t$
        }
    }
    Return NULL
}

\SetKwProg{daconnected}{Function \emph{BackwardBruijnBFS}$(s, d, switches)$}{}{end}
\daconnected{} {
    $Q=t$, $i=0$ \\
    \While{$s \not\in Q$}{
        \ForAll{$z \in Q$} { 
            \If{$z$ has available \da ports in $switches$} {
                Return $(z, port)$
            }
        }
        $i = i + 1$ \\
        $Q =$ all nodes $x$ with $\dist(x,t)=i$
    }
    Return NULL
}
\end{algorithm}

\subsubsection{Centralized scheduling of \da-links}

We consider two centralized algorithms for the scheduling of \da-links. 
Both algorithms use greedy heuristics to add shortcuts (\da links) to the backbone \debruijn network. 
Both algorithms work by periodically (with \emph{period} $\update$) determining the new \da-links based on an estimate of the demand or measurement of the traffic in the network. 
We denote this estimation by a \emph{demand matrix} $D$.
Both algorithms sort the demands in $D$ by decreasing order and for each $(s,t)$ demand in $D$, they try to add a \da link to the network. 
In case the algorithm decides to set a \da link, we assume that the \emph{reconfiguration time} is $\reconf$ and during this time the link is not available for use. 
Additionally when a \da link is set, it stays connected for a \emph{reservation time} of $\resrv$ before it can be replaced, if needed. 

The first algorithm Breadth-First-Search \emph{(BFS)-DA-links}, shown in Algorithm~\ref{alg:dalinks:centralized:BFS}, takes a global perspective.
For each demand $(s,t) \in D$ in decreasing order, it searches for the \emph{shortest} possible path that could be created between $s$ and $t$ by adding a shortcut to the \debruijn backbone. 
From the source $s$, the search follows the path on the static topology part toward $t$ until a node with at least one available \da-port has been found (line~\ref{ln:foundForward}), denote it as $x$. The available \da ports of $x$ are denoted as $switches$ (line~\ref{ln:forward}).  
In turn, a destination-based breadth first search is preformed  until a node $y$ with available \da port in $switches$ is found (line~\ref{ln:backward}). We denote the available port in $y$ as $i$. 
If the new path (with the shortcut) is equal or shorter than the greedy path on the current \system topology using static + \da links (line~\ref{ln:ifBSF}), the algorithm creates a shortcut via a \da link between $x$ and $y$ on the $i$'th spine switch (line~\ref{ln:setBFS}).
Note that initially $x=s$ and $y=t$, but at a later stages of the algorithm it will create integrated multi-hop paths in \system.

The second centralized algorithm, \emph{Greedy-\da-links}, shown in Algorithm~\ref{alg:dalinks:centralized:greedy}, is a simple version of greedy $k$-matchings (known also as $b$-matching for undirected graphs~\cite{gabow2018data}). The algorithm iterates over the requests $(s,t) \in D$ in decreasing order and only connects a direct link between $s$ and $t$ if they have available ports on the same \da switch $i$. 
The greedy matching is a simplified version of the \emph{BFS-DA-links} algorithm, nevertheless, its shortcuts also support integrated multi-hop as before, and a similar version of it is easier to implemented in a distributed way, as we explain next. 

\begin{algorithm}[t]
\caption{Centralized (Greedy) \da links setting}
\label{alg:dalinks:centralized:greedy}
\SetKwProg{daconnected}{Function \emph{Greedy-\da-links}($D$ - Demand Matrix, $k_d$ - number of \da switches)}{}{end}
\daconnected{} {
    $\Delta$=Largest $k_d n$ demands in $D$, sorted by volume \\
        \ForAll{$(s,t) \in \Delta$ from large to small} {
            \If{$s, t$ have available \da ports in switch $i$}{
                Set \da-link $(s,t, i)$ \\
            }
        }
}
\end{algorithm}

\subsubsection{Distributed scheduling of \da-links}

The distributed scheduling algorithm, \emph{Dist\da}, 
shown in Algorithm~\ref{alg:dalinks:distributed}, combines similar approaches as presented in ProjecToR~\cite{projector} and Sirius~\cite{sirius}. 
It implements a distributed, threshold-based greedy $k$ matchings algorithm.
The algorithm is triggered by destination-based elephant detection of flows from a source.
For instance, this can be done in P4 using sketches~\cite{namkung2022sketchlib} as we discuss in more details later.
If any of the destinations detects a source(-ToR) as elephant it checks if it has available \da-ports.
If available, it sends an offer, \emph{PortRequest($ports$)}, to the elephant  source ToR via the static topology part where $ports$ is a list of available ports.
Upon reception, the source/sender checks for an available \da-port on its side. 
If a port is available, it acknowledges the request via a \emph{PortApprove($i$)} message and set the link on port $i$.
If no \da-port is available at the source, the request is declined.
The receiver ToR continues to generate \emph{PortRequests} for other elephants.
An agreed \da-link, i.e., the ports at sending and receiving ToR, is reserved for fixed period of time $\resrv$.
Afterwards, the ports can be assigned to new requests and the circuit might be reconfigured.
However, the circuit is \emph{not} pro-actively torn down but kept alive until an appropriate request arrives.

\begin{algorithm}[t]
\caption{Distributed \da-link scheduling}
\label{alg:dalinks:distributed}
\SetKwProg{daconnected}{Function \emph{Dist\da}() at destination $t$}{}{end}
\daconnected{} {
    Upon detection of \emph{elephant flow} from source $s$ \\
    \If{$t$ has available \da ports}{
        Send \emph{PortRequest(ports)} to node $s$\\
        \If{$s$ reply with \emph{PortApprove(i)}}{
            \da-link $(s,t,i)$ is set (with timeout)\\
        }
    }
}
\SetKwProg{daconnected}{Function \emph{Dist\da}() at source $s$}{}{end}
\daconnected{} {
    Upon \emph{PortRequest(ports)} from destination $t$ \\
    \If{$s$ has available \da ports in $ports$}{
        Send \emph{PortApprove(i)} to node $t$\\
        Set  \da-link $(s,t,i)$ (with timeout)\\
    }
    \Else{Send \emph{DeclineRequest}}
}
\end{algorithm}

We note that while our distributed scheduler is simple, it is effective as we will see next.
We leave the study of more sophisticated schedulers (e.g., based on distributed stable matchings~\cite{projector} or
online algorithms~\cite{kalyanasundaram1993online}) for future work.

\subsection{Implementation and Practical Aspects}

\subsubsection{Implementation and Cost}
As mentioned earlier, we envision  that \system could be implemented  using the Sirius architecture~\cite{sirius}. 
Sirius is also captured by the TMT model, but one of its great advantages is that instead of spine switches, Sirius uses a single layer of $k$ gratings.
The Arrayed Wavelength Grating Routers (AWGR) are simple and \emph{passive} without moving parts and do not consume power. 
Still, each grating diffracts (``forwards'') incoming light from input to output ports, based on the  wavelength, abstractly creating a matching. Reconfiguration is then performed by a physical-layer ToR switch (or directly on servers) equipped with $k$ transceivers containing tunable lasers that can change the wavelength used to carry the data toward the gratings through an optical fiber.

Sirius has been presented as a demand-oblivious architecture which provides fast end-to-end reconfiguration, due to a pre-determined, static schedule that specifies the connectivity at any given fixed-size timeslot. However, Sirius' architecture is in principle also well-suited for demand-aware scheduling, with a slower end-to-end reconfiguration delay. 

As \system differs from Sirius only in the scheduling and routing, the cost and power consumption of \system will be similar to Sirius. 
In~\cite{sirius}, the authors showed that Sirius’ power and cost are about 25\% that of an electrically switched Clos network (ESN). 
That said, unfortunately, a direct comparison of the performance of \system and Sirius is currently not possible as Sirius' simulation code is not available, we therefore concentrate on the comparison to static expander topologies which are also state-of-the-art datacenter networks~\cite{xpander}.

Recently, Cerberus~\cite{sigmetrics22cerberus} which can potentially also be built on the Sirius architecture, demonstrated that using 1-hop \da links (and keeping some demand-oblivious dynamic links as in Sirius or RotorNet~\cite{rotornet}) can increase the network throughput. We, therefore, believe that besides the conceptual contribution of \system, in terms of performance it could enhance any demand-oblivious existing design.

\subsubsection{IP Addressing and LPM Forwarding}\label{sec:design:addressing}

We embed the \debruijn address into the hosts' IP addresses.
Our approach uses IPv4 but can also be implemented using IPv6.
Depending on the number of ports per ToR, a single symbol of the \debruijn address takes one or multiple bits of the IP address:
Thus, the full \debruijn address occupies $s\cdot d$ bits of the IP address.
In order to use LPM to implement the forwarding, we split the IP address into three parts.
The first $p$ bits mark the base network that is assigned to \system.
The following $s'=s\cdot d$ bits identify the ToR by means of the \debruijn address while the remaining bits identify the host/VM inside the rack, i.e., each ToR is assigned a $/(p+s')$ prefix.

For the example of Figure~\ref{fig:design:dynamic_tables}(a), the \debruijn address can directly be mapped to an IP address/prefix and occupies only $3$ address bits.
Using $10.0.0.0/8$ as a base IP prefix, an exemplary forwarding table for ToR $5=011$ is shown in Figure~\ref{fig:design:dynamic_tables}(c).
Following Algorithm~\ref{alg:forwarding:merge} each node can build its IP forwarding table locally based on its ToR neighbors' addresses.
In particular, when a new \da links is established for a node's port and it knows the ToR address of the new neighbor, the forwarding table can be updated locally (without recomputing shortest paths).

\section{Conclusion}\label{sec:conclusion}
To address the limitations and overheads of existing reconfigurable
datacenter networks, we proposed  \system, a simple and flexible architecture 
which supports integrated multi-hop routing and demand-aware links.
\system is work conserving and enables fast topology updates and 
simple control. 
We argued that a realization of \system using a Sirius topology reconfiguration
model may be particularly interesting. 

We understand our work as a next step towards more practical and scalable demand-aware
reconfigurable datacenter networks, and believe that our work opens several interesting
avenues for future research. In particular, while we demonstrated the benefits of greedy and local routing on a \debruijn topology, we believe that our approach is more general and applicable to other network topologies that support greedy local routing. 

\begin{acks}
Research supported by the European Research Council (ERC),  consolidator project Self-Adjusting Networks (AdjustNet), grant agreement No. 864228, Horizon 2020, 2020-2025. The work was also funded by the Deutsche Forschungsgemeinschaft (DFG, German Research Foundation) - 438892507.
\end{acks}

\label{lastmainpage}
{\balance
  \bibliographystyle{ieeetr} 
\bibliography{literature}

\begin{thebibliography}{10}

\bibitem{talk-about}
J.~C. Mogul and L.~Popa, ``What we talk about when we talk about cloud network
  performance,'' {\em SIGCOMM Comput. Commun. Rev. (CCR)}, vol.~42, pp.~44--48,
  Sept. 2012.

\bibitem{li2019hpcc}
Y.~Li, R.~Miao, H.~H. Liu, Y.~Zhuang, F.~Feng, L.~Tang, Z.~Cao, M.~Zhang,
  F.~Kelly, M.~Alizadeh, {\em et~al.}, ``Hpcc: High precision congestion
  control,'' in {\em Proceedings of the ACM Special Interest Group on Data
  Communication}, pp.~44--58, 2019.

\bibitem{bcube}
C.~Guo, G.~Lu, D.~Li, H.~Wu, X.~Zhang, Y.~Shi, C.~Tian, Y.~Zhang, and S.~Lu,
  ``Bcube: a high performance, server-centric network architecture for modular
  data centers,'' {\em ACM SIGCOMM Computer Communication Review}, vol.~39,
  no.~4, pp.~63--74, 2009.

\bibitem{singla2012jellyfish}
A.~Singla, C.-Y. Hong, L.~Popa, and P.~B. Godfrey, ``Jellyfish: Networking data
  centers randomly,'' in {\em Presented as part of the 9th $\{$USENIX$\}$
  Symposium on Networked Systems Design and Implementation ($\{$NSDI$\}$ 12)},
  pp.~225--238, 2012.

\bibitem{jupiter}
A.~Singh, J.~Ong, A.~Agarwal, G.~Anderson, A.~Armistead, R.~Bannon, S.~Boving,
  G.~Desai, B.~Felderman, P.~Germano, {\em et~al.}, ``Jupiter rising: A decade
  of clos topologies and centralized control in google's datacenter network,''
  {\em ACM SIGCOMM computer communication review}, vol.~45, no.~4,
  pp.~183--197, 2015.

\bibitem{AlFares2008}
M.~Al-Fares, A.~Loukissas, and A.~Vahdat, ``A scalable, commodity data center
  network architecture,'' {\em Proc. SIGCOMM Computer Communication Review
  (CCR)}, vol.~38, pp.~63--74, Aug. 2008.

\bibitem{ballani2020sirius}
H.~Ballani, P.~Costa, R.~Behrendt, D.~Cletheroe, I.~Haller, K.~Jozwik,
  F.~Karinou, S.~Lange, K.~Shi, B.~Thomsen, {\em et~al.}, ``Sirius: A flat
  datacenter network with nanosecond optical switching,'' in {\em Proceedings
  of the Annual conference of the ACM Special Interest Group on Data
  Communication on the applications, technologies, architectures, and protocols
  for computer communication}, pp.~782--797, 2020.

\bibitem{zhou2012mirror}
X.~Zhou, Z.~Zhang, Y.~Zhu, Y.~Li, S.~Kumar, A.~Vahdat, B.~Y. Zhao, and
  H.~Zheng, ``Mirror mirror on the ceiling: Flexible wireless links for data
  centers,'' {\em Proc. ACM SIGCOMM Computer Communication Review (CCR)},
  vol.~42, no.~4, pp.~443--454, 2012.

\bibitem{kandula2009flyways}
S.~Kandula, J.~Padhye, and P.~Bahl, ``Flyways to de-congest data center
  networks,'' in {\em Proc. ACM Workshop on Hot Topics in Networks (HotNets)},
  2009.

\bibitem{rotornet}
W.~M. Mellette, R.~McGuinness, A.~Roy, A.~Forencich, G.~Papen, A.~C. Snoeren,
  and G.~Porter, ``Rotornet: A scalable, low-complexity, optical datacenter
  network,'' in {\em Proceedings of the Conference of the ACM Special Interest
  Group on Data Communication}, pp.~267--280, ACM, 2017.

\bibitem{opera}
W.~M. Mellette, R.~Das, Y.~Guo, R.~McGuinness, A.~C. Snoeren, and G.~Porter,
  ``Expanding across time to deliver bandwidth efficiency and low latency,'' in
  {\em 17th $\{$USENIX$\}$ Symposium on Networked Systems Design and
  Implementation ($\{$NSDI$\}$ 20)}, pp.~1--18, 2020.

\bibitem{helios}
N.~Farrington, G.~Porter, S.~Radhakrishnan, H.~H. Bazzaz, V.~Subramanya,
  Y.~Fainman, G.~Papen, and A.~Vahdat, ``Helios: a hybrid electrical/optical
  switch architecture for modular data centers,'' {\em ACM SIGCOMM Computer
  Communication Review}, vol.~41, no.~4, pp.~339--350, 2011.

\bibitem{firefly}
N.~Hamedazimi, Z.~Qazi, H.~Gupta, V.~Sekar, S.~R. Das, J.~P. Longtin, H.~Shah,
  and A.~Tanwer, ``Firefly: A reconfigurable wireless data center fabric using
  free-space optics,'' in {\em ACM SIGCOMM Computer Communication Review},
  vol.~44, pp.~319--330, ACM, 2014.

\bibitem{megaswitch}
L.~Chen, K.~Chen, Z.~Zhu, M.~Yu, G.~Porter, C.~Qiao, and S.~Zhong, ``Enabling
  wide-spread communications on optical fabric with megaswitch,'' in {\em
  Proceedings of the 14th USENIX Conference on Networked Systems Design and
  Implementation}, NSDI'17, (USA), pp.~577--593, USENIX Association, 2017.

\bibitem{quartz}
Y.~J. Liu, P.~X. Gao, B.~Wong, and S.~Keshav, ``Quartz: A new design element
  for low-latency dcns,'' {\em SIGCOMM Comput. Commun. Rev.}, vol.~44,
  pp.~283--294, Aug. 2014.

\bibitem{osa}
K.~{Chen}, A.~{Singla}, A.~{Singh}, K.~{Ramachandran}, L.~{Xu}, Y.~{Zhang},
  X.~{Wen}, and Y.~{Chen}, ``Osa: An optical switching architecture for data
  center networks with unprecedented flexibility,'' {\em IEEE/ACM Transactions
  on Networking}, vol.~22, pp.~498--511, April 2014.

\bibitem{projector}
M.~Ghobadi, R.~Mahajan, A.~Phanishayee, N.~Devanur, J.~Kulkarni, G.~Ranade,
  P.-A. Blanche, H.~Rastegarfar, M.~Glick, and D.~Kilper, ``Projector: Agile
  reconfigurable data center interconnect,'' in {\em Proceedings of the 2016
  ACM SIGCOMM Conference}, pp.~216--229, ACM, 2016.

\bibitem{cthrough}
G.~Wang, D.~G. Andersen, M.~Kaminsky, K.~Papagiannaki, T.~Ng, M.~Kozuch, and
  M.~Ryan, ``c-through: Part-time optics in data centers,'' {\em ACM SIGCOMM
  Computer Communication Review}, vol.~41, no.~4, pp.~327--338, 2011.

\bibitem{splaynets}
S.~Schmid, C.~Avin, C.~Scheideler, M.~Borokhovich, B.~Haeupler, and Z.~Lotker,
  ``Splaynet: Towards locally self-adjusting networks,'' {\em IEEE/ACM
  Transactions on Networking (ToN)}, 2016.

\bibitem{venkatakrishnan2018costly}
S.~B. Venkatakrishnan, M.~Alizadeh, and P.~Viswanath, ``Costly circuits,
  submodular schedules and approximate carath{\'e}odory theorems,'' {\em
  Queueing Systems}, vol.~88, no.~3-4, pp.~311--347, 2018.

\bibitem{schwartz2019online}
R.~Schwartz, M.~Singh, and S.~Yazdanbod, ``Online and offline greedy algorithms
  for routing with switching costs,'' {\em arXiv preprint arXiv:1905.02800},
  2019.

\bibitem{proteus}
A.~Singla, A.~Singh, K.~Ramachandran, L.~Xu, and Y.~Zhang, ``Proteus: a
  topology malleable data center network,'' in {\em Proceedings of the 9th ACM
  SIGCOMM Workshop on Hot Topics in Networks}, p.~8, ACM, 2010.

\bibitem{100times}
M.~Hampson, ``Reconfigurable optical networks will move supercomputerdata 100x
  faster,'' in {\em IEEE Spectrum}, 2021.

\bibitem{fleet}
F.~Douglis, S.~Robertson, E.~Van~den Berg, J.~Micallef, M.~Pucci, A.~Aiken,
  M.~Hattink, M.~Seok, and K.~Bergman, ``Fleet—fast lanes for expedited
  execution at 10 terabits: Program overview,'' {\em IEEE Internet Computing},
  2021.

\bibitem{zhang2021gemini}
M.~Zhang, J.~Zhang, R.~Wang, R.~Govindan, J.~C. Mogul, and A.~Vahdat, ``Gemini:
  Practical reconfigurable datacenter networks with topology and traffic
  engineering,'' {\em arXiv preprint arXiv:2110.08374}, 2021.

\bibitem{sigmetrics22cerberus}
C.~Griner, J.~Zerwas, A.~Blenk, S.~Schmid, M.~Ghobadi, and C.~Avin, ``Cerberus:
  The power of choices in datacenter topology design (a throughput
  perspective),'' in {\em Proc. ACM SIGMETRICS}, 2022.

\bibitem{infocom22lazy}
E.~Feder, I.~Rathod, P.~Shyamsukha, R.~Sama, V.~Aksenov, I.~Salem, and
  S.~Schmid, ``Lazy self-adjusting bounded-degree networks for the matching
  model,'' in {\em Proc. IEEE Conference on Computer Communications (INFOCOM)},
  2022.

\bibitem{ccr18san}
C.~Avin and S.~Schmid, ``Toward demand-aware networking: A theory for
  self-adjusting networks,'' in {\em ACM SIGCOMM Computer Communication Review
  (CCR)}, 2018.

\bibitem{osn21}
M.~N. Hall, K.-T. Foerster, S.~Schmid, and R.~Durairajan, ``A survey of
  reconfigurable optical networks,'' in {\em Optical Switching and Networking
  (OSN), Elsevier}, 2021.

\bibitem{sirius}
H.~Ballani, P.~Costa, R.~Behrendt, D.~Cletheroe, I.~Haller, K.~Jozwik,
  F.~Karinou, S.~Lange, K.~Shi, B.~Thomsen, {\em et~al.}, ``Sirius: A flat
  datacenter network with nanosecond optical switching,'' in {\em Proceedings
  of the Annual conference of the ACM Special Interest Group on Data
  Communication on the applications, technologies, architectures, and protocols
  for computer communication}, pp.~782--797, 2020.

\bibitem{dan}
C.~Avin, K.~Mondal, and S.~Schmid, ``Demand-aware network designs of bounded
  degree,'' in {\em Proc. International Symposium on Distributed Computing
  (DISC)}, 2017.

\bibitem{flexspander}
M.~Y. Teh, Z.~Wu, and K.~Bergman, ``Flexspander: augmenting expander networks
  in high-performance systems with optical bandwidth steering,'' {\em IEEE/OSA
  Journal of Optical Communications and Networking}, vol.~12, no.~4,
  pp.~B44--B54, 2020.

\bibitem{tracecomplexity}
C.~Avin, M.~Ghobadi, C.~Griner, and S.~Schmid, ``On the complexity of traffic
  traces and implications,'' in {\em Proc. ACM SIGMETRICS}, 2020.

\bibitem{benson2010network}
T.~Benson, A.~Akella, and D.~A. Maltz, ``Network traffic characteristics of
  data centers in the wild,'' in {\em Proceedings of the 10th ACM SIGCOMM
  conference on Internet measurement}, pp.~267--280, ACM, 2010.

\bibitem{datacenter_burstiness}
Q.~Zhang, V.~Liu, H.~Zeng, and A.~Krishnamurthy, ``High-resolution measurement
  of data center microbursts,'' in {\em Proceedings of the 2017 Internet
  Measurement Conference}, IMC '17, (New York, NY, USA), pp.~78--85, ACM, 2017.

\bibitem{DBLP:journals/cn/ZouW0HCLXH14}
S.~Zou, X.~Wen, K.~Chen, S.~Huang, Y.~Chen, Y.~Liu, Y.~Xia, and C.~Hu,
  ``Virtualknotter: Online virtual machine shuffling for congestion resolving
  in virtualized datacenter,'' {\em Computer Networks}, vol.~67, pp.~141--153,
  2014.

\bibitem{ancs18}
K.-T. Foerster, M.~Ghobadi, and S.~Schmid, ``Characterizing the algorithmic
  complexity of reconfigurable data center architectures,'' in {\em Proc.
  ACM/IEEE Symposium on Architectures for Networking and Communications Systems
  (ANCS)}, 2018.

\bibitem{taleoftwo}
Y.~Xia, X.~S. Sun, S.~Dzinamarira, D.~Wu, X.~S. Huang, and T.~S.~E. Ng, ``A
  tale of two topologies: Exploring convertible data center network
  architectures with flat-tree,'' in {\em Proceedings of the Conference of the
  ACM Special Interest Group on Data Communication}, SIGCOMM '17, (New York,
  NY, USA), p.~295–308, Association for Computing Machinery, 2017.

\bibitem{francois2005achieving}
P.~Francois, C.~Filsfils, J.~Evans, and O.~Bonaventure, ``Achieving sub-second
  igp convergence in large ip networks,'' {\em ACM SIGCOMM Computer
  Communication Review}, vol.~35, no.~3, pp.~35--44, 2005.

\bibitem{xpander}
S.~Kassing, A.~Valadarsky, G.~Shahaf, M.~Schapira, and A.~Singla, ``Beyond
  fat-trees without antennae, mirrors, and disco-balls,'' in {\em Proceedings
  of the Conference of the ACM Special Interest Group on Data Communication},
  pp.~281--294, ACM, 2017.

\bibitem{stojmenovic2002position}
I.~Stojmenovic, ``Position-based routing in ad hoc networks,'' {\em IEEE
  communications magazine}, vol.~40, no.~7, pp.~128--134, 2002.

\bibitem{scheideler2009distributed}
C.~Scheideler and S.~Schmid, ``A distributed and oblivious heap,'' {\em Proc.
  International Conference on Automata, Languages and Programming (ICALP)},
  pp.~571--582, 2009.

\bibitem{naor2007novel}
M.~Naor and U.~Wieder, ``Novel architectures for p2p applications: the
  continuous-discrete approach,'' {\em ACM Transactions on Algorithms (TALG)},
  vol.~3, no.~3, p.~34, 2007.

\bibitem{fraigniaud2006d2b}
P.~Fraigniaud and P.~Gauron, ``D2b: A de bruijn based content-addressable
  network,'' {\em Theoretical Computer Science}, vol.~355, no.~1, pp.~65--79,
  2006.

\bibitem{kaashoek2003koorde}
M.~F. Kaashoek and D.~R. Karger, ``Koorde: A simple degree-optimal distributed
  hash table,'' in {\em International Workshop on Peer-to-Peer Systems},
  pp.~98--107, Springer, 2003.

\bibitem{louri1995optical}
A.~Louri and H.~Sung, ``Optical binary de bruijn networks for massively
  parallel computing: design methodology and feasibility study,'' {\em Applied
  optics}, vol.~34, no.~29, pp.~6714--6722, 1995.

\bibitem{deBruijn}
N.~G. De~Bruijn, ``A combinatorial problem,'' in {\em Proc. Koninklijke
  Nederlandse Academie van Wetenschappen}, vol.~49, pp.~758--764, 1946.

\bibitem{leighton2014introduction}
F.~T. Leighton, {\em Introduction to parallel algorithms and architectures:
  Arrays{\textperiodcentered} trees{\textperiodcentered} hypercubes}.
\newblock Elsevier, 2014.

\bibitem{hall1935}
P.~Hall, ``On representatives of subsets,'' {\em Journal of the London
  Mathematical Society}, vol.~s1-10, no.~1, pp.~26--30, 1935.

\bibitem{durr2016flat}
F.~D{\"u}rr, ``A flat and scalable data center network topology based on de
  bruijn graphs,'' {\em arXiv preprint arXiv:1610.03245}, 2016.

\bibitem{gabow2018data}
H.~N. Gabow, ``Data structures for weighted matching and extensions to
  b-matching and f-factors,'' {\em ACM Transactions on Algorithms (TALG)},
  vol.~14, no.~3, pp.~1--80, 2018.

\bibitem{namkung2022sketchlib}
H.~Namkung, Z.~Liu, D.~Kim, V.~Sekar, P.~Steenkiste, G.~Liu, A.~Li, C.~Canel,
  A.~A. Philip, R.~Ware, {\em et~al.}, ``Sketchlib: Enabling efficient
  sketch-based monitoring on programmable switches,'' NSDI, 2022.

\bibitem{kalyanasundaram1993online}
B.~Kalyanasundaram and K.~Pruhs, ``Online weighted matching,'' {\em Journal of
  Algorithms}, vol.~14, no.~3, pp.~478--488, 1993.

\end{thebibliography}
}
\label{lastpage}
\end{document}